\documentclass[sigconf,vvarbb, natbib=false]{acmart}

\usepackage[utf8]{inputenc}
\usepackage[T1]{fontenc}
\usepackage{algorithm}
\usepackage{xspace}
\usepackage{algorithmic}
\usepackage{enumitem}
\usepackage[kw]{pseudo}
\pseudoset{
    compact,
    label=\scriptsize\sffamily\color{gray}\arabic*,
    ctfont=\sffamily\color{gray},
    ct-left=\ctfont{\# },
    ct-right=,
    line-height=1.2,
    indent-length=3ex,
    left-margin=0pt,
}

\usepackage[frozencache]{minted}
\newminted{python}{escapeinside=!!}
\input{minted-arxiv/algol_nu.pygstyle}

\usepackage{tikz}
\usetikzlibrary{graphs}
\usetikzlibrary{graphs.standard}
\usetikzlibrary{shapes.geometric}
\usetikzlibrary{shapes.symbols}
\usetikzlibrary{topaths} 
\tikzset{
  treenode/.style={
    circle,
    fill=black,
    inner sep=.2em,
  },
  subtree/.style={
    isosceles triangle,
    shape border rotate=90,
    isosceles triangle apex angle=60,
    inner sep=0em,
    minimum size=2.5em,
    anchor=apex,
    draw
  },
  st/.style={
    signal, draw,
    signal to=west,
    anchor=west,
  },
}

\usepackage[style = acmnumeric, isbn = false, backend = biber, sortcites=true, alldates=year]{biblatex}
\IfFileExists{maths.bib}{\addbibresource{maths.bib}}{}

\newlist{inputoutput}{description}{1}
\setlist[inputoutput]{labelwidth=4em, font=\upshape\scshape\sffamily, labelsep=0pt, leftmargin=4em}

\newcommand{\lm}{\operatorname{lm}}

\newcommand{\sfrak}{\mathfrak{s}}

\newcommand{\lcm}{\operatorname{lcm}}

\newcommand{\K}{\mathbb{K}}
\newcommand{\N}{\mathbb{N}}

\newcommand{\mop}[1]{\operatorname{#1}}
\newcommand{\dom}{\operatorname{dom}}
\renewcommand{\st}{\ \middle|\ }

\newcommand{\Ffour}{\textsc{\texorpdfstring{F\textsubscript{4}}{F4}}\xspace}
\newcommand{\Ffive}{\textsc{\texorpdfstring{F\textsubscript{5}}{F5}}\xspace}
\newcommand{\ideal}{\operatorname{\mathbb{I}}}

\newcommand{\crit}{\operatorname{crit}}

\newcommand{\setof}[2]{\left\{#1\;\middle|\;#2\right\}}
\newcommand{\set}[1]{\left\{#1\right\}}
\newcommand{\pmap}{\rightharpoonup}

\theoremstyle{acmplain}
\newtheorem{theorem}{Theorem}
\newtheorem{lemma}[theorem]{Lemma}
\newtheorem{corollary}[theorem]{Corollary}
\newtheorem{proposition}[theorem]{Proposition}

\theoremstyle{acmdefinition}
\newtheorem{definition}[theorem]{Definition}

\copyrightyear{2026}
\acmYear{2026}
\setcopyright{cc}
\setcctype{by}
\acmConference[ISSAC '26]{51st International Symposium on Symbolic and Algebraic Computation}{July 13--17, 2026}{Oldenburg, Germany}
\acmBooktitle{51st International Symposium on Symbolic and Algebraic Computation (ISSAC '26), July 13--17, 2026, Oldenburg, Germany}
\acmDOI{10.1145/3815436.3815473}
\acmISBN{979-8-4007-2595-1/2026/07}

\begin{document}

\title[A data structure for monomial ideals
with applications to signature Gröbner bases]{A data structure for monomial ideals\\
with applications to signature Gröbner bases}

\author{Pierre Lairez}
\orcid{0000-0003-3756-0151}
\affiliation{
  \institution{Inria, Université Paris-Saclay}
  \city{Palaiseau}
  \country{France}
}
\email{pierre.lairez@inria.fr}

\author{Rafael Mohr}
\orcid{0009-0000-9869-7122}
\affiliation{
  \institution{KU Leuven, Dep. of Computer Science}
  \city{Leuven}
  \country{Belgium}
}
\email{rafaeldavid.mohr@kuleuven.be}

\author{Théo Ternier}
\orcid{0009-0007-7508-2863}
\affiliation{
  \institution{Inria, Université Paris-Saclay}
  \city{Palaiseau}
  \country{France}
}
\email{theo.ternier@inria.fr}

\begin{abstract}
  We introduce \emph{monomial divisibility diagrams} (MDDs), a data structure for monomial ideals that supports insertion of new generators and fast membership tests.
  MDDs stem from a canonical tree representation by maximally sharing equal subtrees, yielding a directed acyclic graph.
  We establish basic complexity bounds for membership and insertion, and study empirically the size of MDDs.
  As an application, we integrate MDDs into the signature Gröbner basis implementation of the Julia package \emph{AlgebraicSolving.jl}.
  Membership tests in monomial ideals are used to detect some reductions to zero, and the use of MDDs leads to substantial speed-ups compared to the existing representation by lists of generators with divmasks.
\end{abstract}



\maketitle

\section{Introduction}

A Gröbner basis computation breaks down into two kinds of work:
operations on monomials (multiplications, comparisons, divisibility tests),
and arithmetic operations in the ground field. The latter occur when reducing
new polynomials by previously computed relations.
A polynomial that reduces to zero does not contribute to the output.
The family of algorithms stemming from the \Ffive algorithm \parencite{Faugere_2002,EderRoune_2013,EderFaugere_2017,GaoVolnyWang_2016,Lairez_2024} relies on \emph{signatures}, monomials attached to polynomials, to detect many reductions to zero before they happen.
This typically decreases the number of arithmetic operations, at the cost of more monomial operations.

As it turns out, one source of reductions to zero in Gröbner basis computations is the trivial commutator relation $fg - gf = 0$.
Reductions to zero arising from this relation can be avoided using signatures by what was initially called the \Ffive criterion, whose verification requires checking whether a given signature lies in a monomial ideal.
Applying this criterion may yield a substantial
number of monomial ideal membership tests.  For
example, computing a Gröbner basis of the benchmark ideal
\emph{eco-14} with the \Ffive implementation provided by
\emph{AlgebraicSolving.jl} yields 117\,861\,802 membership tests.  In
this example, membership queries not only dominate the symbolic
workload but also account for a significant part of the whole
computation, unlike in non-signature Gröbner basis algorithms,
such as \Ffour \cite{faugere1999}, where the workload
typically stems mostly from arithmetic operations.  This raises the
question: can one support both efficient insertion and fast membership
tests for monomial ideals?

The most obvious data structure for a monomial ideal is a list of generators.
An ideal $I \subseteq \K[x_1,\dots,x_n]$ is stored as $(m_1,\dots,m_r)$, with monomials represented by exponent vectors in~$\N^n$.
Insertion amounts to appending to the list, and thus takes time~$O(1)$, ignoring maintenance of minimality.
A membership query for a monomial $q$ tests whether $m_i \mid q$ for each generator $m_i$.
Assuming exponent comparisons cost $O(1)$ per variable, each divisibility test costs $O(n)$, so one query costs $O(rn)$.

With a different space--time tradeoff, we may represent a monomial ideal~$I\subseteq \K[x_1,\dotsc,x_n]$ generated in partial degrees at most~$D$ by an~$n$-dimensional boolean array with $(D+1)^n$ entries, indexed by tuples $(i_1,\dots,i_n)$ 
with $0 \le i_k \le D$, such that~$A[i_1]\dotsb[i_n]$ is \emph{true} if and only if~$x_1^{i_1} \dotsb x_n^{i_n}$ is in~$I$.
This representation supports membership queries in time~$O(n)$: given a monomial $q$ with exponent vector $(i_1,\dotsc,i_n)$, we compute $j_k = \min(i_k, D)$ and read~$A[j_1]\dotsb[j_n]$.
However, inserting a new generator can be costly: to preserve the ideal property, all entries corresponding to multiples of~$q$ must be set to \emph{true}, which costs~$\Theta(D^n)$ in the worst case.
The space consumption also makes this naive representation impractical beyond small examples.

\paragraph*{Contribution}
We propose \emph{monomial divisibility diagrams (MDDs)} as a data structure for monomial ideals.
An MDD for an ideal~$I$ is a compact encoding of the corresponding $n$-dimensional array~$A$ considered above as a directed acyclic graph.
If the partial degrees of the generators are bounded by~$D$, membership in a monomial ideal can be decided in time~$O(n \log D)$ using an MDD, independently of the number of generators.
MDDs are cousins of \emph{binary decision diagrams (BDDs)}, which represent boolean functions~$\left\{ 0, 1\right\}^n \to \left\{ 0,1 \right\}$ \parencite{Knuth_2011}.
We implemented MDDs in \emph{AlgebraicSolving.jl} and obtained significant speed-ups over the previous divmask implementation.

\paragraph*{Previous work}
Divisibility masks (\emph{divmasks}) are compact bit patterns attached to generators that
discard many potential divisibility tests cheaply \parencite[e.g.][]{RouneStillman_2012}.
They are used, for example, in \emph{AlgebraicSolving.jl} and \emph{msolve} \parencite{BerthomieuEderSafeyElDin_2021}.
The masks need to be recomputed periodically as the ideal evolves.
Divmasks improve the practical speed of the naive algorithm for ideal membership, but the complexity is still linear in the number of generators.

Several tree data structures for monomial ideals have been proposed to speed up membership tests and related operations.
\Textcite{Milowski_2004} introduced \emph{monomial trees}, which are prefix trees built on exponent vectors of generators.
These trees can accelerate membership tests when many monomials share similar exponent patterns, because parts of the tree can be skipped, but may also degenerate, for example when exponents in the first variable are all distinct.

There have been different definitions of \emph{Janet trees} \parencite{GerdtBlinkovYanovich_2001,HashemiOrthSeiler_2022}
for representing Janet bases of monomial ideals by prefix trees of exponent vectors, providing an efficient membership test \parencite[Theorem~1]{GerdtBlinkovYanovich_2001}.
The generalization of Janet trees to minimal Janet-like bases \parencite{GerdtBlinkov_2005,HashemiOrthSeiler_2023} is essentially equivalent to our definition of ideal tree (Section~\ref{sec:ideal-trees}).
With respect to this line of work, our contribution is twofold:
first, a theoretical framework independent of the notion of involutive division;
second, a systematic use of maximal sharing, which makes the data structure usable in cases where a Janet basis would not even fit on a hard disk.

\textcite{GaoZhu_2008} also use prefix trees of exponent vectors to compute irreducible decompositions of monomial ideals.
Their \emph{MinMerge step} can be interpreted as an on-the-fly computation of a Janet-like basis.

\emph{Support trees}, as implemented by \textcite{Malkin_2007} in \emph{4ti2}, provide another tree representation.
At each node, the tree branches according to whether a given variable divides the monomial, and monomials are stored in the leaves.
This has proved efficient in the context of toric ideals, but the data structure degenerates
when the supports of the exponent vectors are too similar.

\textcite{RouneStillman_2012,RouneStillman_2012a} propose the use of \emph{kd-trees} for monomial ideals.
Their experiments show clear speed-ups over simple lists of generators and over monomial trees on a range of benchmarks.
However, kd-trees can become unbalanced after many insertions, so they must be rebuilt periodically, and finding good splitting criteria for a given distribution of monomials is not straightforward. The worst-case complexity for membership is linear in the number of generators, even when the tree is balanced.

Lastly, membership testing in a monomial ideal is a special case of \emph{range searching} (see \cite{Agarwal_2017} for a review, and \cite{Dube_1993} for membership testing).
The \emph{range tree} data structure \parencite{Bentley_1979} can store monomials as points in~$\mathbb{R}^n$. It can be constructed in~$O(r (\log r)^{n-1})$ operations and answers membership queries in~$O( (\log r)^n )$ operations.
There is a whole line of research on different space-time tradeoffs.
MDDs instead favor query time: their worst-case space complexity is exponential, but they remain compact in many applications considered here. This approach is similar to that of binary decision diagrams, which have found many applications in spite of their exponential worst-case space complexity.

\section{Monomial ideals as trees}

\subsection{Monomial ideals}

A monomial ideal in~$n$ variables
is an ideal~$I\subseteq\mathbb{Q}[x_1,\dotsc,x_n]$ generated by monomials.
It is well known that a polynomial~$f$ is in~$I$ if and only if every monomial in~$f$
is divisible by one of the monomials generating~$I$ \parencite[\S2.4]{CoxLittleOShea_2015}.
In particular, a monomial ideal is entirely determined by the set of monomials contained in it.
Identifying a monomial~$x_1^{\alpha_1} \dotsm x_n^{\alpha_n}$ with its exponent vector~$(\alpha_1,\dotsc,\alpha_n) \in \mathbb{N}^n$,
we henceforth view a monomial ideal as a subset~$I\subseteq\mathbb{N}^n$ satisfying
$I + \mathbb{N}^n = I$.
For~$E\subseteq \mathbb{N}^n$ we write $\langle E\rangle = E + \mathbb{N}^n$ for the ideal generated by~$E$.
By Gordan-Dickson's Lemma \parencite[Theorem 2.5]{CoxLittleOShea_2015}, every monomial ideal admits a finite generating set.

In~$\mathbb{N}^0$ there are exactly two monomial ideals, namely $\varnothing$ and~$\mathbb{N}^0$.
For~$n > 0$, we can describe the monomial ideals in~$\mathbb{N}^n$ in terms of monomial ideals in~$\mathbb{N}^{n-1}$.
Let $I\subseteq \mathbb{N}^n$ be a monomial ideal, and let~$e \geq 0$.
We define the \emph{$e$-th quotient ideal of~$I$} as
\[ I/e = \left\{ \alpha \in \mathbb{N}^{n-1} \st \alpha \cdot e \in I \right\}, \]
where~$\alpha \cdot e$ denotes the concatenation~$(\alpha_1,\dotsc,\alpha_{n-1},e)$.
In terms of ideals in~$\mathbb{Q}[x_1,\dotsc,x_n]$,
the definition is equivalent to
\[ I/e = \left( I : x_n^e \right) \cap \mathbb{Q}[x_1,\dotsc,x_{n-1}]. \]
Note that~$I$ is entirely determined by the quotients~$I/e$ since
\begin{equation}
  \label{eq:2}
  I = \bigcup_{e\geq 0}\,(I/e)\cdot e.
\end{equation}


\subsection{Ideal trees}
\label{sec:ideal-trees}

Given a set~$\mathcal{T}$, let~$\mathbb{N}\pmap \mathcal{T}$ be the set of partial functions from~$\mathbb{N}$ to~$\mathcal{T}$ with finite domain.
The domain of a partial function~$t$ is denoted~$\dom(t)$. For~$e \in \dom(t)$, we write~$t(e)$ for the value of the partial function~$t$ at~$e$.

\begin{definition}
  Let~$\bot$ be a set that is not a partial function from~$\mathbb{N}$ to any other set.
  The set~$\mathcal{T}_n$ of $n$-trees is defined recursively as
  \begin{itemize}
    \item $\mathcal{T}_0 = \left\{ \bot \right\}$,
    \item $\mathcal{T}_n = \left\{ t : \mathbb{N} \pmap \mathcal{T}_{n-1} \st \dom(t) \neq \varnothing \right\}$, for~$n > 0$.
  \end{itemize}
\end{definition}
In other words, an $n$-tree is a finite tree,
in which all leaves lie at depth exactly~$n$.
Moreover, the edges are labeled by natural integers.
There is a one-to-one correspondence between~$n$-trees and non\-empty finite subsets of~$\mathbb{N}^n$ by defining
$\mop{set}(\bot) = \mathbb{N}^0$ and
\[ \mop{set}(t)
  = \bigcup_{e\in\dom(t)} \mop{set}(t(e))\cdot e. \]
Equivalently, $\mop{set}(t)$ corresponds to the set of paths in~$t$, from the root to a leaf.
This is the classical construction of \emph{prefix trees} (except that we read the words backward). This yields an interesting data structure to represent finite subsets of~$\mathbb{N}^n$.
In this section, we develop an analogous construction for monomial ideals.

To this end, we define the \emph{ideal represented by an $n$-tree~$t$} as
the monomial ideal generated by~$\mop{set}(t)$, namely
$\ideal(t) = \langle \mop{set}(t) \rangle$.
Equivalently, we have, for~$n=0$, $\ideal(\bot) = \mathbb{N}^0$, and, for~$n > 0$,
\begin{equation}
  \label{eq:1}
  \ideal(t) = \bigcup_{e\in \dom(t)}\bigcup_{k\geq 0}\:\:\ideal(t(e))\cdot (e+k).
\end{equation}

\begin{definition}
  An $n$-tree is \emph{ideal} if either~$n = 0$, or~$n> 0$ and the following conditions hold:
  \begin{itemize}
    \item for all~$e \in \dom(t)$, $t(e)$ is ideal;
    \item for all~$e, f\in \dom(t)$,
          $e < f \implies \ideal(t(e)) \subset \ideal(t(f))$,
          where~$\subset$ denotes strict inclusion.
  \end{itemize}
\end{definition}

\begin{figure}[tb]
  \centering

  \includegraphics[height=3cm, trim={1cm 1.5cm 1cm 1cm}]{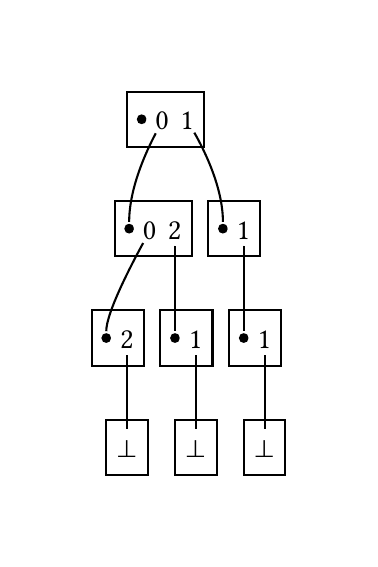}\quad\includegraphics[height=3cm, trim={1cm 1.5cm 1cm 1cm}]{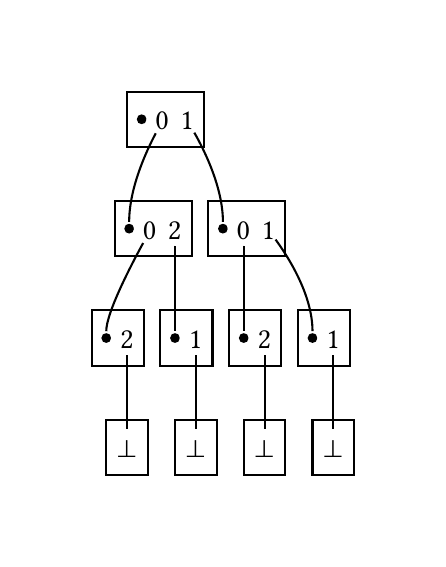}
  
  \caption{Two trees representing the monomial ideal  $\langle xyz, x^2, xy^2 \rangle$. The first level corresponds to the variable~$z$, the second to~$y$, and the last to~$x$.
    The left tree is not an ideal tree because~$\langle x^2, xy^2\rangle$ is not included in~$\langle xy \rangle$.
    The right tree is an ideal tree. Note the extra path corresponding to the monomial~$zx^2$.}
  \label{fig:small}
\end{figure}

Figure~\ref{fig:small} shows an example.

An ideal~$n$-tree~$t$ provides an efficient data structure for deciding membership in~$\ideal(t)$.
Indeed, consider the case~$n > 0$ and write~$\alpha = \alpha'\cdot e \in \mathbb{N}^n$.
By definition, $\alpha \in \ideal(t)$ if and only if~$\alpha' \in \ideal(t(f))$ for some~$f \in \dom(t)$ with $f \leq e$.
For a general~$n$-tree this would require checking all such~$f$.
When~$t$ is ideal, the ideals~$\ideal(t(f))$ form a strictly increasing chain, so it suffices to test
$\alpha' \in \ideal\bigl(t(f)\bigr)$ for the largest~$f$ not exceeding~$e$.
This leads to the procedure~\fn{contains} (Algorithm~\ref{algo:contains}).
The following lemma summarizes the argument.

\begin{lemma}\label{lem:quotient}
  Let~$n > 0$ and let~$t$ be an ideal~$n$-tree.
  For any~$e \geq 0$:
  \begin{enumerate}[label=(\roman*)]
    \item\label{item:empty-below-min} If~$e < \min(\dom(t))$, then~$\ideal(t)/e = \varnothing$.
    \item\label{item:quotient} Otherwise,
          $\ideal(t)/e = \ideal\bigl(t(f)\bigr)$,
          where~$f$ is the largest element of~$\dom(t)$ which does not exceed~$e$.
  \end{enumerate}
\end{lemma}

\begin{proof}
  Equation~\eqref{eq:1} implies
  \[ \ideal(t)/e = \bigcup_{f\in \dom(t),\ f \leq e} \ideal(t(f)). \]
  If~$e < \min(\dom(t))$, then the union is empty
  and the first point follows.
  Otherwise, the defining property of ideal trees shows that
  the union is contained in the term~$\ideal(t(f))$ corresponding to the largest such~$f$, which yields the second point.
\end{proof}

In the same way as there is a unique~$n$-tree representing a given nonempty subset of~$\mathbb{N}^n$,
there is a unique ideal~$n$-tree representing a given nonempty ideal of~$\mathbb{N}^n$.

\begin{theorem}\label{thm:ideal-tree}
  The map~$t \in \mathcal{T}_n \mapsto \ideal(t) \subseteq \mathbb{N}^n$ induces a bijection between
  ideal $n$-trees and nonempty monomial ideals in~$\mathbb{N}^n$.
\end{theorem}

\begin{proof}
  We argue by induction on~$n$. For~$n=0$, the claim is immediate because there is exactly one ideal~$0$-tree, namely~$\bot$, and exactly one nonempty ideal in~$\mathbb{N}^0$, namely~$\mathbb{N}^0$.

  Assume $n>0$ and let $I \subseteq \mathbb{N}^n$ be a nonempty monomial ideal.
  By the ascending chain condition the sequence
  $I/0 \subseteq  I/1 \subseteq I/2 \subseteq \dotsb$
  eventually stabilizes.
  Since~$I$ is nonempty, the stabilized value is nonempty as well.
  Choose indices $0 \leq f_1 < \dotsb < f_r$ (and put $f_{r+1} = \infty$) such that
  \begin{enumerate}[label=(\roman*)]
    \item \label{item:efirst} $I/e = \varnothing$ for all~$0\leq e < f_1$,
    \item \label{item:ecorrect} $I/e = I/{f_i}$ whenever~$f_i \leq e < f_{i+1}$, and
    \item \label{item:emin} $I/{f_i} \neq I/{f_{i+1}}$ for all~$1\leq i < r$.
  \end{enumerate}
  By the induction hypothesis, each $I/{f_i}$ is represented by a unique ideal $(n-1)$-tree. Define $t$ by $\dom(t) = \{f_1,\dotsc,f_r\}$ and $t(f_i)$ is the ideal tree representing~$I/f_i$.
  The strict inclusions among the $I/{f_i}$, by~\ref{item:emin}, ensure that $t$ is an ideal $n$-tree.
  Lemma~\ref{lem:quotient} yields $\ideal(t)/e = I/e$ for every $e \geq 0$, hence $\ideal(t) = I$, by~\eqref{eq:2}. This proves surjectivity.

  For injectivity, let $s$ be another ideal $n$-tree with $I = \ideal(s) = \ideal(t)$. Suppose there exists $e \in \dom(s) \setminus \dom(t)$ and choose such an $e$ minimal. If $e < f_1$, then $\ideal(s)/e = \varnothing$, by~\ref{item:efirst}, contradicting $\ideal(s(e)) \neq \varnothing$. Otherwise $f_i < e < f_{i+1}$ for a unique $i$. Minimality of~$e$ implies~$f_i \in \dom(s)$.
  Lemma~\ref{lem:quotient} and \ref{item:ecorrect} give
  $\ideal(s(e)) = I/e = I/f_i = \ideal(s(f_i))$.
  This contradicts the strict monotonicity of~$s$. Hence $\dom(s) \subseteq \dom(t)$, and the reverse inclusion follows from similar arguments.

  Lemma~\ref{lem:quotient} gives $\ideal(s(e)) = \ideal(t(e))$ for every $e \in \dom(t) = \dom(s)$. The induction hypothesis then yields $s(e) = t(e)$, so~$s = t$.
\end{proof}

\begin{algorithm}[tp]
  \caption{Membership in an ideal tree}
  \label{algo:contains}
  \raggedright
  \begin{inputoutput}
    \item[input] an ideal $n$-tree $t$, with~$n>0$, and~$e \in \mathbb{N}$
    \item[output] either an ideal $(n-1)$-tree $s$ such that~$\ideal(t)/e = \ideal(s)$, \\or~$\varnothing$ if~$\ideal(t)/e = \varnothing$.
  \end{inputoutput}
  \begin{pseudo}
def \fn{quo}(t, e):\\+
  if $\min(\dom(t)) > e$:\\+
    return $\varnothing$\\-
  $f\gets $ \tn{largest element of $\dom(t)$ that does not exceed~$e$}\\
  return $t(f)$\\-
\end{pseudo}

\bigskip
\begin{inputoutput}
  \item[input] an ideal~$n$-tree~$t$, and $\alpha \in \mathbb{N}^n$
  \item[output] \kw{True} if~$\alpha \in \ideal(t)$, \kw{False} otherwise
\end{inputoutput}
\begin{pseudo}
def \fn{contains}(t, \alpha):\\+
  if $n = 0$:\\+
    return True\\-
  $s\gets \fn{quo}(t, \alpha_n)$\\
  if $s$ = $\varnothing$:\\+
    return False\\-
  return \fn{contains}(s, (\alpha_1,\dotsc,\alpha_{n-1}))
  \end{pseudo}
\end{algorithm}

A nonempty monomial ideal~$I\subseteq \mathbb{N}^n$ now comes with two distinguished generating sets.
The first is the familiar \emph{minimal} generating set.
The second is~$\mop{set}(t)$, where~$t$ is the unique ideal tree representing~$I$.
This generating set has appeared before as the \emph{minimal Janet-like basis of~$I$} in the context of involutive divisions and related algorithms for computing Gröbner bases \parencite{GerdtBlinkov_2005}.
The notion of ideal trees provides an equivalent but conceptually simpler description, independent of involutive divisions.

\subsection{Insertion in an ideal tree}
\label{sec:insertion-an-ideal}

\begin{figure}[p]
  \centering
    \tikzstyle{every node}=[font=\small]

  \rlap{before insertion}\begin{tikzpicture}[scale=0.6, x=.5cm, y=0.7cm, out=-90, in=90, looseness=0.6]
    \path[anchor=base]  (-3, 0) node (bullet) {$\bullet$} -- ++(1, 0) node (a) {$a$}  -- ++(1, 0) node (b) {$b$} -- ++(1, 0) node (c) {$c$}  -- ++(1,-.3) node (leftb) {};
    \draw (bullet.west) + (-.5, .5) rectangle (leftb);

    \node (Ia) at (-6, -2) {$t(a)$}; 
    \draw (a) edge (Ia) ;
    \node (Ib) at (-1, -2) {$t(b)$}; 
    \draw (b) edge (Ib) ;
    \node (Ic) at (4, -2) {$t(c)$}; 
    \draw (c) edge (Ic) ;
  \end{tikzpicture}

  \bigskip
  \rlap{case $x < a$}
  \begin{tikzpicture}[scale=0.6, x=.5cm, y=0.8cm, out=-90, in=90, looseness=0.6]
    \path[anchor=base]  (-5, 0) node (bullet) {$\bullet$} -- ++(1, 0) node (x) {$x$} -- ++(1, 0) node (a) {$a$}  -- ++(1, 0) node (b) {$b$} -- ++(1, 0) node (c) {$c$}  -- ++(1,-.3) node (leftb) {};
    \draw (bullet.west) + (-.5, .5) rectangle (leftb);

    \node (Ix) at (-13, -2) {$\fn{singleton}(\alpha')$}; 
    \draw (x) edge[out=-120] (Ix) ;
    \node (Ia) at (-6, -2) {$\fn{insert}(t(a), \alpha')$}; 
    \draw (a) edge (Ia) ;
    \node (Ib) at (1, -2) {$\fn{insert}(t(b), \alpha')$}; 
    \draw (b) edge (Ib) ;
    \node (Ic) at (8, -2) {$\fn{insert}(t(c), \alpha')$}; 
    \draw (c) edge[out=-60] (Ic) ;
  \end{tikzpicture}

    \bigskip
  \rlap{case $a < x < b$}
  \begin{tikzpicture}[scale=0.6, x=.5cm, y=0.8cm, out=-90, in=90, looseness=0.6]
    \path[anchor=base]  (-5, 0) node (bullet) {$\bullet$}  -- ++(1, 0) node (a) {$a$} -- ++(1, 0) node (x) {$x$} -- ++(1, 0) node (b) {$b$} -- ++(1, 0) node (c) {$c$}  -- ++(1,-.3) node (leftb) {};
    \draw (bullet.west) + (-.5, .5) rectangle (leftb);

    \node (Ia) at (-13, -2) {$t(a)$}; 
    \draw (a) edge[out=-120] (Ia) ;
    \node (Ix) at (-6, -2) {$\fn{insert}(t(a), \alpha')$}; 
    \draw (x) edge (Ix) ;
    \node (Ib) at (1, -2) {$\fn{insert}(t(b), \alpha')$}; 
    \draw (b) edge (Ib) ;
    \node (Ic) at (8, -2) {$\fn{insert}(t(c), \alpha')$}; 
    \draw (c) edge[out=-60] (Ic) ;
  \end{tikzpicture}

  \bigskip
  \rlap{case $x = b$}
  \begin{tikzpicture}[scale=0.6, x=.5cm, y=0.8cm, out=-90, in=90, looseness=0.6]
    \path[anchor=base]  (-3, 0) node (bullet) {$\bullet$} -- ++(1, 0) node (a) {$a$}  -- ++(1, 0) node (b) {$b$} -- ++(1, 0) node (c) {$c$}  -- ++(1,-.3) node (leftb) {};
    \draw (bullet.west) + (-.5, .5) rectangle (leftb);

    \node (Ia) at (-6, -2) {$t(a)$}; 
    \draw (a) edge (Ia) ;
    \node (Ib) at (-1, -2) {$\fn{insert}(t(b), \alpha')$}; 
    \draw (b) edge (Ib) ;
    \node (Ic) at (6, -2) {$\fn{insert}(t(c), \alpha')$}; 
    \draw (c) edge (Ic) ;
  \end{tikzpicture}

  \caption{\normalfont Insertion of a monomial~$\alpha = \alpha'\cdot x$ in an ideal tree~$t$ with~$\dom(t) = \left\{ a, b, c \right\}$, with~$a < b<c$. Pruning of redundant subtrees after insertion is not shown.}
  \bigskip
  \label{fig:insertion}
\end{figure}
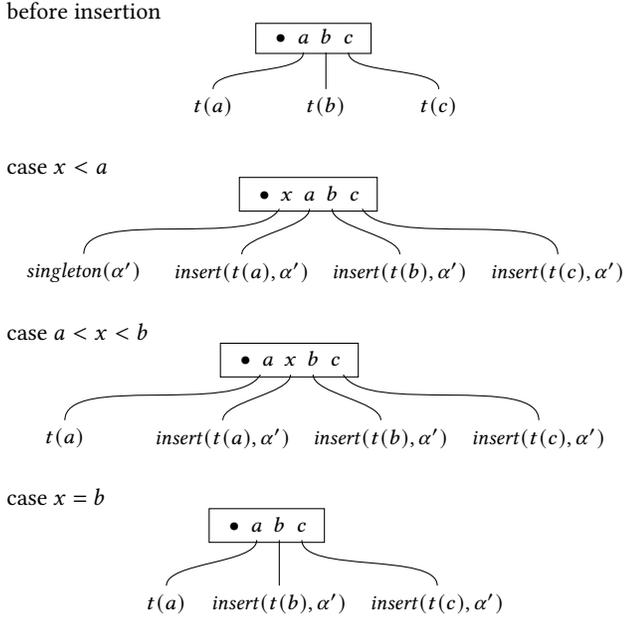

\begin{algorithm}[p]
  \caption{Insertion in an ideal tree.}
  \label{algo:insertion}
  \raggedright
  \begin{inputoutput}
    \item[input] $\alpha \in \mathbb{N}^n$
    \item[output] the ideal tree representing $\langle \alpha \rangle$
  \end{inputoutput}
  \begin{pseudo}
def \fn{singleton}(\alpha):\\+
  if $\fn{len}(\alpha) = 0$: 
    return $\{ \}$\\
  $\alpha'\cdot x \gets \alpha$\\
  $t \gets \left\{ x \mapsto \fn{singleton}(\alpha') \right\}$\\
  return $t$\\-
  \end{pseudo}
  \begin{inputoutput}
    \item[input] an ideal tree~$t$ (or $\varnothing$), and~$\alpha \in \mathbb{N}^n$
    \item[output] the ideal tree representing~$\ideal(t) + \langle \alpha\rangle$
  \end{inputoutput}
  \begin{pseudo}
def \fn{insert}(t, \alpha):\\+
  if $t = \varnothing$:
    return \fn{singleton}(\alpha)\\
  if $\fn{len}(\alpha) = 0$:
    return $\{ \}$\\
  $s\gets \{ \}$\\
  $\alpha'\cdot x \gets \alpha$\\
  for $e\in \dom(t)$:\\+\
    if $e < x$:\\+
      $s(e) \gets t(e)$\\-
    elif $e > x$:\\+
      $s(e) \gets \fn{insert}(t(e), \alpha')$\\--
  $s(x) \gets \fn{insert}(\fn{quo}(t, x), \alpha')$\\
  \tn{remove any~$f \in \dom(s)$ s.t.~$s(f) = s(e)$ for some~$e < f$}\\
  return $s$\\-
  \end{pseudo}
\end{algorithm}

For a monomial~$\alpha = (\alpha_1, \dotsc, \alpha_n) \in \mathbb{N}^n$, let~$\fn{singleton}(\alpha)$ denote the ideal tree representing the ideal~$\langle\alpha\rangle$, whose existence and uniqueness is guaranteed by Theorem~\ref{thm:ideal-tree}.
This is the tree\\
\begin{center}
\begin{tikzpicture}
  \node[treenode, label={\small root}] (root) {};
  \draw (root) -- node[above] {$\alpha_n$} ++(1, 0) node[treenode] {} -- node[above] {$\alpha_{n-1}$} ++(1, 0) node[treenode] (a) {} ;
  \draw (a) -- ++(.5, 0) node (a1) {};
  \draw[dotted] (a1) -- ++(1, 0) node (a2) {};
  \draw (a2) -- ++(.5, 0) node[treenode] {} -- node[above] {$\alpha_{1}$} ++(1, 0) node[treenode] {} ;
\end{tikzpicture}.
\end{center}

For an ideal~$n$-tree~$t$ and a monomial~$\alpha \in \mathbb{N}^n$, let~$\fn{insert}(t, \alpha)$
denote the unique ideal tree representing the ideal~$\ideal(t) \cup \langle \alpha \rangle$.
The case~$n=0$ is trivial (since there is a single ideal~$0$-tree, namely~$\bot$),
so let us consider the case~$n > 0$ and write~$\alpha = \alpha' \cdot x$.
For any~$e < x$,
\[ \left( \ideal(t) + \langle \alpha \rangle \right)/e = \ideal(t)/e, \]
while for any~$e \geq x$,
\[ \left( \ideal(t) + \langle \alpha \rangle \right)/e = \ideal(t)/e + \langle \alpha' \rangle, \]
Therefore, we consider the tree~$s$
defined by
\begin{enumerate}
  \item $\dom(s) = \dom(t) \cup \left\{ x \right\}$,
  \item $s(e) = t(e)$ for all~$e < x$ in~$\dom(s)$,
  \item $s(e) = \fn{insert}(t(e), \alpha')$ for all~$e > x$ in~$\dom(s)$,
  \item $s(x) = \fn{singleton}(\alpha')$ if~$x < \min(\dom(t))$,
  \item $s(x) = \fn{insert}(t(f), \alpha')$ if~$x \geq \min(\dom(t))$, where~$f$ denotes the largest element of~$\dom(t)$ that does not exceed~$x$.
\end{enumerate}
This is illustrated in Figure~\ref{fig:insertion}.
By the discussion above, $\ideal(s) = \ideal(t) \cup (\alpha + \mathbb{N}^n)$, however~$s$ may not be ideal:
the tree~$s$ satisfies~$\ideal(s(e)) \subseteq \ideal(s(f))$ for any~$e, f\in \dom(s)$ with~$e < f$,
but the inclusion may be an equality.
So we remove from~$\dom(s)$ all such elements~$f$, which do not change~$\ideal(s)$, and we obtain the description of~\fn{insert}(t, \alpha).
Algorithm~\ref{algo:insertion} summarizes the construction.

\section{Monomial divisibility diagrams}

\begin{figure*}[htbp]
  \centering

  \includegraphics[width=\linewidth,  trim={1cm 1.5cm 1cm 1cm}]{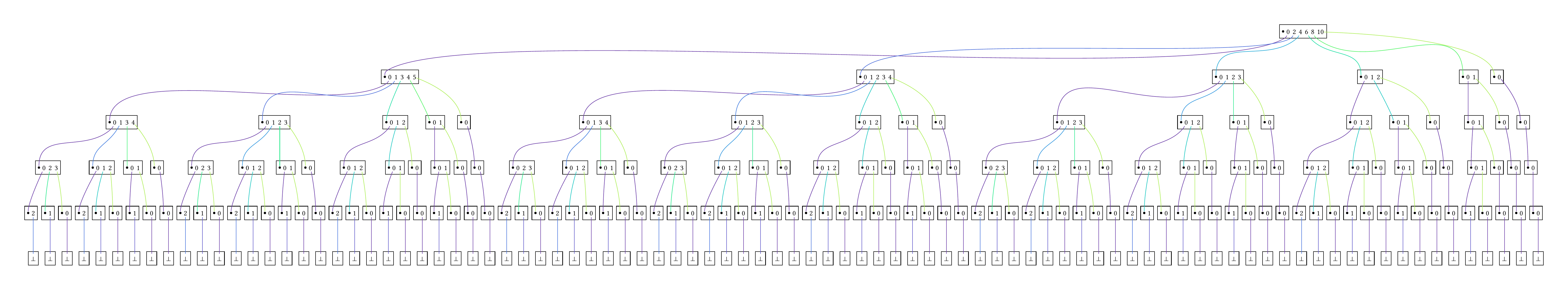}
  \caption{The ideal tree of the monomial ideal generated in~$\mathbb{N}^5$ by the leading monomials of the 55 elements of the Gröbner basis of five generic equations of degree~3, 3, 3, 3, and~2.
  }
  \label{fig:ideal-tree}
\end{figure*}

\begin{figure}[h]
  \centering

  \includegraphics[height=7cm, trim={1cm 1.5cm 1cm 1cm}]{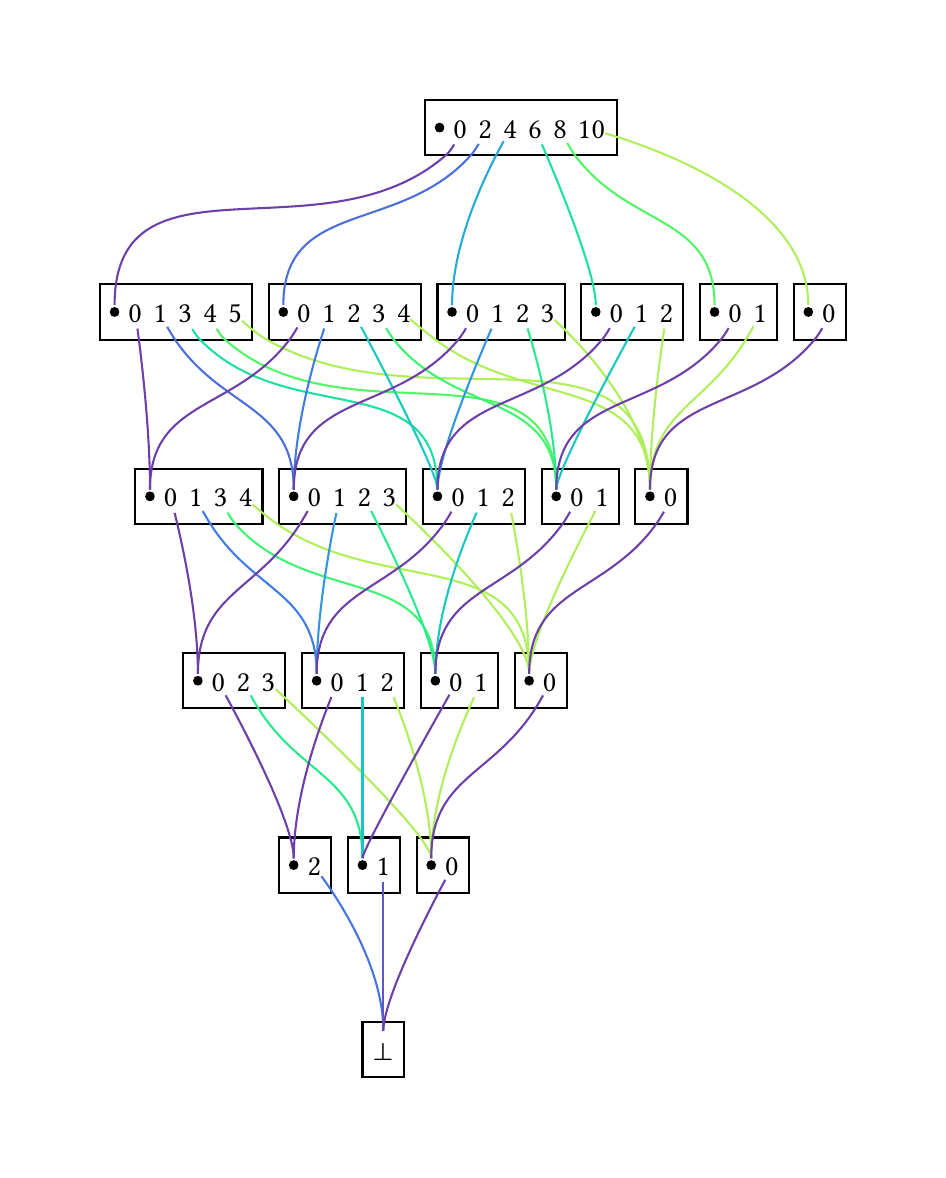}
  \caption{\normalfont The monomial divisibility diagram of the ideal tree of Figure~\ref{fig:ideal-tree}.}
  \label{fig:mdd}
\end{figure}

Direct implementations of ideal trees may face a crushing growth in the number of nodes,
hindering the use of the data structure for practical purposes.
In many cases this growth is only apparent, because the same subtrees occur repeatedly in the ideal tree.
By sharing those common subtrees we obtain a representation of the tree as a directed acyclic graph.
We define a \emph{monomial divisibility diagram (MDD)} as the directed acyclic graph obtained from an ideal tree by maximally sharing its common subtrees. Figure~\ref{fig:ideal-tree} illustrates an ideal tree and Figure~\ref{fig:mdd} displays the corresponding MDD. Figure~\ref{fig:eco14} illustrates another MDD.
Mathematically speaking, ideal trees and MDDs encode the same information, but the compactness of MDDs makes them suitable for practical applications.

\subsection{Formal definition}

Let~$I \subseteq \mathbb{N}^n$ be a nonempty monomial ideal
and let~$t$ be its ideal tree.
Let~$S$ be the set of subtrees of~$t$: it is the smallest subset of~$\mathcal{T}_0 \cup \dotsb \cup \mathcal{T}_n$
which contains~$t$ and contains all the children of all its elements.
The \emph{monomial divisibility diagram} of~$I$ is the directed acyclic graph with vertex set~$S$,
with an edge~$(t, t(e))$ labeled~$e$ for any~$t\in S$ and~$e\in \dom(t)$.

The algorithms for insertion and membership are the same as for ideal trees (Algorithms~\ref{algo:contains} and~\ref{algo:insertion})
but they can be implemented more efficiently using memoization. The critical part is to maintain the maximal sharing property: whenever two subtrees are equal, they must be represented by the same vertex of the MDD.

\subsection{Algorithms and implementation}
\label{sec:algor-impl}
The implementation used in our benchmarks is in Julia, but we present the same ideas here in Python, since concrete code describes the data structure and the
implementation of maximal sharing through hash consing more precisely than pseudocode.
\input{minted-arxiv/listing1.pygtex}
\input{minted-arxiv/listing2.pygtex}

An edge is a pair consisting of an integer \texttt{label} and a destination node \texttt{to};
a node is a tuple of edges, sorted by increasing \texttt{label}. 
Below, we implement the membership test using a binary search, although the nodes are typically small enough that it offers no advantage over linear search.
The call \texttt{contains(t, a, i)} tests whether the monomial with exponent vector~\texttt{a[0:i]} is in the ideal represented by~\texttt{t}.

\input{minted-arxiv/listing3.pygtex}

The following procedure creates a singleton tree for the monomial~\texttt{a[0:i]}, as defined in Section~\ref{sec:insertion-an-ideal}.
\input{minted-arxiv/listing4.pygtex}

To create MDDs efficiently, it is important to decide the structural equality of two nodes in constant time.
A well-known technique is \emph{hash consing} \parencite[e.g.][]{Goto_1974,FilliatreConchon_2006}: we maintain the invariant that \emph{two nodes are structurally equal if and only if they are physically equal} (maximal sharing).
For this purpose, we maintain a hash table \texttt{\_universe} that maps tuples of edges to nodes.
When creating a new node from a tuple of edges, we check whether this tuple is already present in \texttt{\_universe}; if so, we return the existing node instead of allocating a new one.
To ensure that unused nodes are reclaimed by the garbage collector, we store weak references in this dictionary.
We therefore redefine the \texttt{\_\_new\_\_} method of the class \texttt{Node} as follows.
\input{minted-arxiv/listing5.pygtex}

This invariant allows us to implement constant-time equality checks and hashing. In Python, \texttt{id(self)} refers to the address in memory of the object~\texttt{self}.
\input{minted-arxiv/listing6.pygtex}

For example, the following \texttt{dedup} procedure takes a list of edges and removes any edge that points to the same node as the previous edge.
\input{minted-arxiv/listing7.pygtex}

Hash consing not only minimizes memory usage; it also enables efficient memoization. Indeed, maximal sharing ensures that equal subtrees are represented by the same node, so each subproblem needs to be treated only once.
For example, to compute the size of the ideal tree corresponding to an MDD, we use the following procedure.
\input{minted-arxiv/listing8.pygtex}

The same principle applies when implementing the insertion of a monomial in an MDD (Algorithm~\ref{algo:insertion}).\footnote{A useful optimization, omitted here, is the following: if inserting a monomial~$m$ into~$t(e)$ leaves~$t(e)$ unchanged, which can be tested in constant time, then~$m\in\ideal(t(e))$ and hence~$m\in\ideal(t(f))$ for all~$f > e$. Thus some recursive calls can be skipped.}
\input{minted-arxiv/listing9.pygtex}

\subsection{Complexity}

To study the complexity of operations on MDDs, we consider the following size parameters.
Given the MDD of an ideal $I\subseteq \mathbb{N}^n$,
its \emph{size} is the number of edges;
its \emph{width} is the maximum number of nodes at a given depth;
and its \emph{branching degree} is the maximum outdegree of a node.
For example, the size of the MDD in Figure~\ref{fig:mdd} is~52, its width is~6, and its branching degree is~6.

Let~$I \subseteq \mathbb{N}^n$ be a nonempty ideal, represented by its MDD of size~$N$, width~$w$, and branching degree~$\delta$.
We easily check that $N \leq n \delta w$, $w \leq \delta^{n-1}$, and $\delta \leq e+1$ for any~$e$ such that~$I$ is generated by monomials with exponents at most~$e$.

\begin{proposition}
  The cost of testing membership of some~$\alpha \in \mathbb{N}^n$ in~$I$ is
  $O(n \log \delta)$.
\end{proposition}

\begin{proof}
  There are at most~$n$ recursive calls, and each call performs a binary search on an array with at most~$\delta$ elements.
\end{proof}

It is easy to adapt the representation of nodes to obtain an $O(n)$ membership test:
instead of storing a sparse array, we may use a dense array. For example, a node
corresponding to an ideal tree~$t$ with~$\dom(t) = \{ 2, 3, 6 \}$ is represented by an array of pairs
such as $\left[ (2, t(2)), (3, t(3)), (6, t(6)) \right]$,
but we could instead use an array of subtrees, such as
$\left[ \varnothing, \varnothing, t(2), t(3), t(3), t(3), t(6) \right]$.
However, the memory overhead of the dense representation is unbounded.
We therefore advocate for the sparse representation, even though the dense representation may work well in many practical cases.

\begin{proposition}
  The MDD of $I + \langle \alpha \rangle$ can be computed from the MDD of~$I$ with~$O(N)$ operations on average, with a randomized algorithm.
  Moreover, $\operatorname{width}(I + \langle \alpha\rangle) \leq 2 \operatorname{width}(I) + 1$.
\end{proposition}

\begin{proof}
  Algorithm \fn{insert} (implemented with memoization as in \S\ref{sec:algor-impl})
  visits every node at most once, and the work at a node is bounded by the number of outgoing edges (including the deduplication in~\verb|Node.__new__|, on average, implementing hash tables with universal hashing \parencites{CarterWegman_1979}[\S11.3.3]{CormenLeisersonRivestStein_2009}).
  So the total work is bounded by the total number of edges.

  As for the width, we check that every subtree at depth~$k$ of the ideal tree of~$I + \langle \alpha \rangle$
  is one of the following: a subtree at depth~$k$ of the ideal tree of~$I$;
  $\fn{insert}(t, (\alpha_1,\dotsc,\alpha_{n-k}))$, where~$t$ is a subtree at depth~$k$ of~$I$;
  or~$\fn{singleton}(\alpha_{1}\dotsb\alpha_{n-k})$.
\end{proof}

\begin{corollary}\label{coro:exponential-bound}
  If~$I$ is generated by~$r$ monomials, then~$w \leq 2^r - 1$.
\end{corollary}

This exponential bound may seem alarming, but it does not capture the sharing observed in practice.
We do observe an exponential blowup in the MDD size in some cases, especially for ideals generated by random monomials, but also much more favorable cases; see~\S\ref{sec:exper-data-size}.
In forthcoming work, we establish the bound~$w \leq r$
for \emph{Borel-fixed} ideals, which typically appear in Gröbner basis computations \parencite{Green_1998}.

\subsection{Experimental data on the size of MDDs}\label{sec:exper-data-size}

\begin{figure}[tp]
  \centering

  \includegraphics[width=\linewidth]{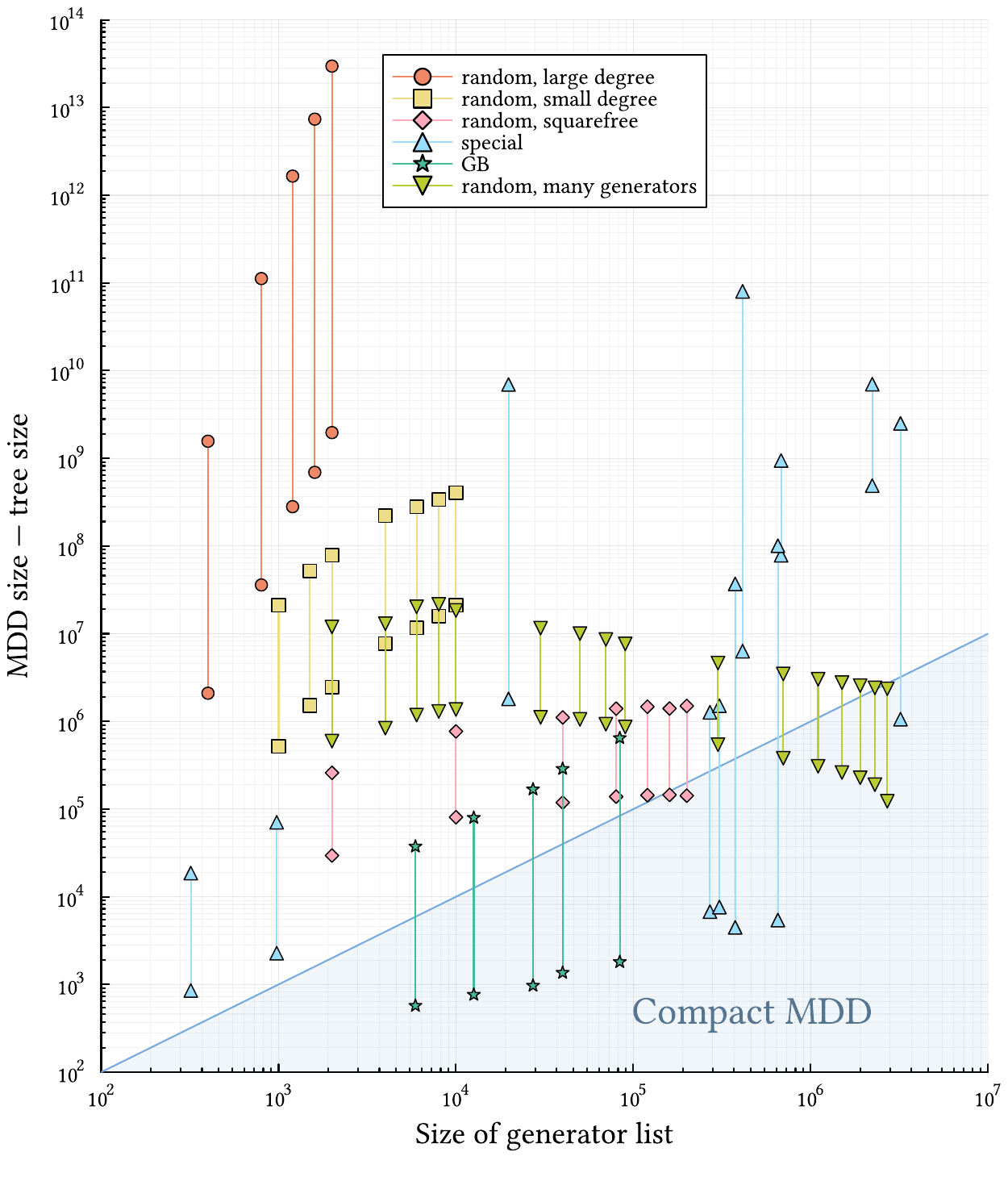}
  
  \caption{Size of the ideal tree and of the MDD across benchmark families.
    Each pair of points corresponds to a monomial ideal. The abscissa of a point is the size of the list of generators, in machine words, that is, the number of generators times the number of variables.
    The ordinate of the lower point of a pair is the size, in machine words, of the MDD in sequential form, that is, twice the number of edges plus the number of nodes.
    The ordinate of the upper point is the size of the ideal tree in the same sequential form.
    The shaded area indicates where the ordinate is smaller than the abscissa.}
  \label{fig:sizeplot}
\end{figure}

It is difficult to quantify \emph{a priori} the benefit of sharing, but in practice the reduction in size is by several orders of magnitude.
As an extreme example, consider the monomial ideal~$I_{n,d} \subseteq \mathbb{N}^n$
generated by all monomials of total degree~$d$.
This ideal is minimally generated by~$\binom{d+n-1}{n-1}$ monomials,
but the width of its MDD is only~$d+1$: the MDD is exponentially more compact than the generator list.

For a practical experiment, we consider six benchmark families of monomial ideals, of which the first four come from the dataset compiled by \textcite[supplementary data]{Roune_2009}:
\begin{description}[font=\normalfont\itshape]
  \item[random, large degree.] Monomial ideals in 10~variables generated by 40 to 200 random monomials; exponents are uniformly distributed in $\left\{ 0,\dotsc,30\,000 \right\}$.
  \item[random, small degree.] Monomial ideals in 10~variables generated by 100 to~1\,000 random monomials; exponents are uniformly distributed in $\left\{ 0,\dotsc,10 \right\}$. 
  \item[random, squarefree.] Monomial ideals in 20~variables generated by 100 to~10\,000 random monomials; exponents are uniformly distributed in $\left\{ 0, 1 \right\}$.
  \item[random, many generators.] Monomial ideals in 10~variables generated by 1000 to~290\,000 random uniformly distributed monomials of total degree exactly 12.
  \item[special.] A collection coming from various problems \parencite[\S 7.1]{Roune_2009}.
  \item[GB.] Ideals generated by leading monomials of grevlex Gröbner bases for the systems Katsura-11, -12, -13, and Eco-14, -15.  
\end{description}

Figure~\ref{fig:sizeplot} plots the size of the MDD of each of the monomial ideals in the dataset as a function of the size of the list of generators.
To appreciate the benefit of maximal sharing, the figure also represents the size of the corresponding ideal tree.
Size is measured in terms of an idealized memory footprint.
Since the actual memory usage of a MDD depends on many implementation-dependent factors,
we associate to each data structure a sequential representation as a flat list of machine words.
For the list of generators, this is the concatenation of all exponent vectors.
For MDDs, we concatenate the representations of the distinct subtrees, where a subtree $t$ with domain~$\left\{ a_1,\dotsc,a_r \right\}$ is encoded as the sequence $r, a_1, p_1,\dotsc, a_r, p_r$, where~$p_i$ is the index of the node corresponding to~$t(a_i)$ in this list.
Thus each node of the MDD contributes one word and each edge two words. The tree size is measured similarly.

We draw several conclusions from this experiment.
First, maximal sharing is what makes the data structure usable: it reduces the tree size by at least an order of magnitude. In some instances, the tree size exceeds a terabyte.

Second, in many cases the MDD is larger than the generator list.
Then, depending on the workload, the construction cost of the MDD may exceed its benefit. In any case, the MDD may still be useful for other computations in commutative algebra, such as Hilbert series or primary decomposition of monomial ideals, which will be developed in forthcoming work.

Last, there are also many instances of the opposite behavior, where the MDD is more compact than the generator list.
Most importantly, this is the case for monomial ideals coming from Gröbner bases.
This is related to the weakly-grevlex \parencite{MorenoSocias_2003} or Borel-fixed structure \parencite{Green_1998} that these monomial ideals often exhibit,
exemplified by the ideal~$I_{n,d}$ discussed above.
The phenomenon will be investigated in future work from a theoretical point of view.

\section{Application to Signature Gröbner Bases}

\subsection{Theory}

We now present an application of MDDs to the computation of signature Gröbner
  bases (sGB). Signature Gröbner bases have their origin in the \Ffive
algorithm \cite{Faugere_2002} and aim to avoid reductions to zero in Buchberger's 
algorithm, which do not contribute to the output but still take time to compute. 
These reductions to zero are avoided by introducing additional monomial ideal 
membership tests which our data structure is designed to do better.
To illustrate how we apply MDDs in sGB computations we
present here a simplified and specific version of an sGB
algorithm (see \cite{EderFaugere_2017,Lairez_2024} for the full theory).

Let $R:= \mathbb{K}[x_1,\dots,x_n]$ be a
polynomial ring over a field $\mathbb{K}$ and let $\prec$ be a monomial order on
the set of $n$-variate monomials. For $f\in R$ nonzero, $\lm(f)$ denotes the
largest monomial occurring in the support of $f$ with respect to the
monomial order $\prec$, and, for~$S\subseteq R\setminus \set{0}$,
$\lm(S)$ denotes~$\setof{\lm(p)}{p\in S}$.
As in the rest of the paper, we identify the set of monomials in $R$ with~$\mathbb{N}^n$.
Finally, for a set $S\subseteq R\setminus \set{0}$ and a nonzero polynomial $p\in R$, write
\[\crit(S,p):= \setof{\frac{\lcm(\lm(p),\lm(s))}{\lm(p)}}{s\in S}.\]

The key component of any Gröbner basis computation is the {\em
  reduction} process, which works as follows (see
\parencite[\S2.3]{CoxLittleOShea_2015} for further details): Given a
finite set $S\subset R\setminus \set{0}$ and a nonzero polynomial $p\in R$, let $t$ be a term in
$p$ such that there is $s\in S$ and a term $u$ with
$t = u\cdot \lm(s)$.  We replace $p$ by $p - us$ and repeat this
process, looking for further terms in the result that are divisible by
an element in $\lm(S)$, until no term in the resulting polynomial lies
in the monomial ideal generated by $\lm(S)$. During this reduction
process we can use an MDD representing $\langle \lm(S) \rangle$ to check whether
there exists an element $s\in S$ whose leading term divides $t$ and,
only if the answer is yes, look for the actual element $s$ using a
linear search. In practice, this can save many iterations
over $\lm(S)$.

We now apply MDDs specifically to sGB computations.
Suppose that we are given a Gröbner basis
$G\subset \mathbb{K}[x_1,\dots,x_n]$ with respect to $\prec$ as well as a
polynomial $f\in \mathbb{K}[x_1,\dots,x_n]$. In its simplest form, the
\Ffive{} algorithm aims to compute a Gröbner basis for the ideal
generated by $G$ and $f$. We can then compute a Gröbner
basis of an ideal generated by $f_1,\dots,f_r\in R$ iteratively,
introducing the~$f_i$ one after the other.

We now give a (too optimistic) optimization of Buchberger's algorithm for computing
a Gröbner basis for $\langle G,f\rangle$:
\begin{enumerate}
\item Initialize $P:=\{(\mathbf{0},f)\}\subset \mathbb{N}^n\times R$ and
  $G_{\text{new}}:= G$. Here $\mathbf{0}\in \mathbb{N}^n$ is the all zeros
        vector, to be understood as the monomial $1\in R$.
  \item \label{step:loop} While~$P$ is not empty:
        \begin{enumerate}
          \item \label{step:reduce} Choose some tuple $(u,p)\in P$, remove it from
                $P$, and reduce $u p$ by $G_{\text{new}}$ to obtain a polynomial~$q$.
          \item \label{step:spair} If~$q\neq 0$, add to $P$ the elements $(v, q)$ for each
                $v\in \crit(G_{\text{new}},q)$, add $q$ to $G_{\text{new}}$, and go
                to step \ref{step:loop}.
        \end{enumerate}
\end{enumerate}
Step \ref{step:spair} does not construct the S-pairs explicitly; this is performed implicitly in step~\ref{step:reduce}.
However, this is flawed because when we reduce~$up$, the polynomial~$p$ is already in~$G_\text{new}$, so~$up$ reduces trivially to~0.
Signatures are introduced into the process to fix the issue.

To obtain an sGB algorithm from the loop above we make the following
three modifications: First, we associate a \emph{signature} to every element of~$G_\text{new}$.
This is a monomial computed inductively as follows:
\begin{itemize}
  \item the polynomial $f$ itself gets the signature $\mathbf{0}$;
  \item a polynomial~$q$ obtained in step~\ref{step:reduce} by reducing some~$u p$ gets the signature~$\sfrak(q) = u \sfrak(p)$.
\end{itemize}

The second modification affects the reduction process in step
\ref{step:reduce}: When reducing an element $up$ we allow only
monomial multiples of elements in $G$ as reducers as well as elements
$vq$ where $q\in G_{\text{new}}\setminus G$, $v\in \mathbb{N}^n$ and
$v\sfrak(q) \prec u\sfrak(p)$. When making these restrictions in step~\ref{step:reduce},
we say that we are {\em regularly reducing} $(u,p)$
by $G_{\text{new}}$.

The third and final modification is that in step \ref{step:reduce} we
always choose a tuple $(u,p) \in P$ such that $u\sfrak(p)$ is minimal.

The key algorithmic benefit of the introduction of signatures now lies
in the following lemma:

\begin{lemma}[Lemmas 4.2 and 6.1 in \cite{EderFaugere_2017}]
  \label{lem:sig}
  Any two pairs
  $(u,p)$ and $(v,q)\in P$ with $u\sfrak(p) = v\sfrak(q)$ regularly reduce to
  the same polynomial by $G_{\text{new}}$. Moreover, an
  element $(u,p)\in P$ regularly reduces to zero by $G_{\text{new}}$ if one
  of the following two statements holds:
  \begin{enumerate}
  \item An element $(v,q)$ that was previously in $P$ regularly reduced
    to zero by $G_{\text{new}}$ and $v\sfrak(q)$ divides $u\sfrak(p)$.
  \item We have $u\sfrak(p)\in \langle \lm(G) \rangle$.
  \end{enumerate}
\end{lemma}

The two criteria in Lemma \ref{lem:sig} are called the {\em syzygy}
and {\em Koszul} criteria in the literature and involve performing a
monomial membership test. Applying them during the run of a
signature-based Gröbner basis algorithm is where MDDs can
be utilized, as described in Algorithm~\ref{algo:sgb}.

\begin{algorithm}[tp]
  \caption{sGB computation with MDDs.}
  \label{algo:sgb}
  \raggedright
  \begin{inputoutput}
    \item[input] A Gröbner basis $G$ w.r.t. a monomial order $\prec$ on $R$, a polynomial $f\in R$.
    \item[output] A Gröbner basis of the ideal $\langle G \cup \{f\}\rangle$ w.r.t. $\prec$.
  \end{inputoutput}
  \begin{pseudo}
    def \fn{sGB}(G,f):\\+
    $G_{\text{new}}\gets G$;\quad  
    $P \gets \{(\mathbf{0}, f)\}$; \quad
    $t \gets $ \tn{the MDD of} $\langle \lm(G)\rangle $\\
    while $P\neq \varnothing$:\\+
    $(u,p)\gets $ \tn{an element of minimal signature in} $P$ \label{ln:choose}\\
    $P\gets P \setminus \setof{(v,q)}{v\sfrak(q) = u\sfrak(p)}$\\
    if not $\fn{contains}(t,u\sfrak(p))$:\\+
    $q\gets $ \tn{result of regularly reducing $up$ by $G_{\text{new}}$}\\
    if $q = 0$:\\+
    $t\gets \fn{insert}(t,u\sfrak(p))$\\-
    else:\\+
    $P \gets P \cup \setof{(v,q)}{v\in \crit(G_{\text{new}},q)}$\\
    \tn{$G_{\text{new}}\gets G_{\text{new}}\cup \{q\}$}\\---
    return $G_{\text{new}}$
  \end{pseudo}
\end{algorithm}

  A natural question regarding Algorithm \ref{algo:sgb} is which element of minimal
  signature to choose in line \ref{ln:choose} if there are multiple choices with the same
  signature. In sGB algorithms, this is handled by a procedure called {\em rewriting}
  which involves further monomial divisibility checks \cite[Section 7]{EderFaugere_2017} that are
  slightly more involved than checking membership in a monomial ideal. How to combine
  this rewriting procedure with MDDs will be the subject of future work.

\subsection{Benchmarks}

\begin{table*}[!ht]
\centering
\scriptsize
\begin{tabular*}{\textwidth}{l@{\extracolsep{\fill}}rrrrrrrrrrrr}
\toprule
 & \multicolumn{2}{c}{\textit{AlgebraicSolving.jl}} & \multicolumn{7}{c}{\textit{AlgebraicSolving.jl with MDD}} & \multicolumn{2}{c}{speed-up} \\
\cmidrule(lr){2-3} \cmidrule(lr){4-10} \cmidrule(lr){11-12}
name & time & symbolic part & time & symbolic part & time memb. & nb memb. & time ins. & nb of insertions & nb of nodes & overall & symbolic\\
\midrule
cp\_d\_3\_n\_6\_p\_2 & 5\,s & 82\,$\%$ & 2\,s & 61\,$\%$ & 0.02 & 211\,k & 0.05 & 1561 & 47 & $\times$\,2.1 & $\times$\,2.8\\
cp\_d\_3\_n\_7\_p\_7 & 5\,s & 73\,$\%$ & 3\,s & 59\,$\%$ & 0.04 & 456\,k & 0.01 & 1020 & 47 & $\times$\,1.6 & $\times$\,1.9\\
cp\_d\_3\_n\_8\_p\_2 & 964\,s & 63\,$\%$ & 559\,s & 37\,$\%$ & 0.78 & 4.7\,M & 3.54 & 11343 & 85 & $\times$\,1.7 & $\times$\,2.9\\
cp\_d\_4\_n\_8\_p\_8 & 89\,s & 59\,$\%$ & 61\,s & 47\,$\%$ & 0.31 & 3.6\,M & 0.06 & 2967 & 58 & $\times$\,1.5 & $\times$\,2.1\\
cp\_d\_4\_n\_9\_p\_9 & 1822\,s & 50\,$\%$ & 1380\,s & 34\,$\%$ & 2.15 & 29\,M & 0.35 & 8624 & 74 & $\times$\,1.3 & $\times$\,1.9\\
cyclic7 & 1\,s & 98\,$\%$ & 0.8\,s & 98\,$\%$ & 0.02 & 240\,k & 0.02 & 750 & 43 & $\times$\,1.6 & $\times$\,1.6\\
cyclic8 & 68\,s & 95\,$\%$ & 32\,s & 89\,$\%$ & 0.50 & 4.1\,M & 0.36 & 3881 & 80 & $\times$\,2.1 & $\times$\,2.3 \\
eco14 & 12118\,s & 90\,$\%$ & 4150\,s & 68\,$\%$ & 5.79 & 118\,M & 0.92 & 19241 & 173 & $\times$\,2.9 & $\times$\,3.8\\
gametwo7 & 1625\,s & 95\,$\%$ & 1336\,s & 94\,$\%$ & 2.13 & 41\,M & 0.95 & 12234 & 75 & $\times$\,1.2 & $\times$\,1.2\\
katsura11 & 4\,s & 86\,$\%$ & 2\,s & 79\,$\%$ & 0.04 & 388\,k & 0.06 & 1024 & 81 & $\times$\,1.8 & $\times$\,2.0\\
katsura12 & 27\,s & 83\,$\%$ & 14\,s & 70\,$\%$ & 0.17 & 1.5\,M & 0.27 & 2048 & 103 & $\times$\,2.0 & $\times$\,2.4\\
katsura13 & 201\,s & 82\,$\%$ & 107\,s & 69\,$\%$ & 0.71 & 5.9\,M & 1.14 & 4096 & 125 & $\times$\,1.9 & $\times$\,2.2\\
noon7 & 10\,s & 87\,$\%$ & 3\,s & 60\,$\%$ & 0.08 & 985\,k & 0.02 & 1392 & 71 & $\times$\,3.1 & $\times$\,4.4\\
noon8 & 154\,s & 82\,$\%$ & 44\,s & 34\,$\%$ & 0.63 & 7.4\,M & 0.10 & 3843 & 101 & $\times$\,3.5 & $\times$\,8.4\\
noon9 & 78\,s & 99\,$\%$ & 83\,s & 99\,$\%$ & 0.58 & 7.2\,M & 0.10 & 3743 & 71 & $\times$\,0.9 & $\times$\,0.9\\
yang1 & 11569\,s & 99\,$\%$ & 6280\,s & 99\,$\%$ & 22.27 & 138\,M & 14.41 & 39188 & 947 & $\times$\,1.8 & $\times$\,1.8\\
\bottomrule
\end{tabular*}
\medskip
\caption{Comparison of \emph{AlgebraicSolving.jl} and \emph{AlgebraicSolving.jl} with MDD (our Julia implementation).
    \emph{time}: computation time of the signature Gröbner basis in seconds;
    \emph{symbolic part}: percentage of computation time spent in the symbolic computation part of the algorithm;
    \emph{time memb.}: computation time in seconds of the membership tests;
    \emph{nb memb.}: number of membership tests during the algorithm;
    \emph{time ins.}: computation time of the monomial divisibility diagram in seconds;
    \emph{nb of insertions}: number of insertions in the monomial divisibility diagram;
    \emph{nb of nodes}: number of nodes of the monomial divisibility diagram at the end of the algorithm.}
\label{table:algsolvingbenchmarks}
\end{table*}

We integrated MDDs into the implementation of signature Gröbner bases in the Julia package \emph{AlgebraicSolving.jl}. The source code is available online and distributed under the GPLv3 license.\footnote{The modified version of \emph{AlgebraicSolving.jl} is available at \url{https://github.com/algebraic-solving/AlgebraicSolving.jl}.}
In the new version, membership tests are performed using \emph{contains} (Algorithm~\ref{algo:contains}), whereas the old version uses linear search over the generator list, with the divmask optimization.

Table~\ref{table:algsolvingbenchmarks} reports the computation times obtained for sGB computations on several benchmark systems selected from the \emph{msolve} library \cite{msolveexamples}. These include the \emph{cp} systems, which consist of generic polynomial systems, as well as classical benchmark families.

An sGB algorithm, such as Algorithm \ref{algo:sgb}, has two
parts: a purely \emph{symbolic} part (divisibility tests on
signatures, finding reducers, computation of $S$-pairs, etc.) and an
\emph{arithmetic} part (the arithmetic operations involved in reducing
polynomials). The use of MDDs affects only the symbolic part.

Table~\ref{table:algsolvingbenchmarks} reports on the impact of MDDs on the \emph{total} computation time,
as well as the symbolic part of the computation time.
In most examples, the total computation time is significantly reduced by using MDDs,
with the exception of \emph{noon9}, which operates on monomial ideals with very few generators.
With MDDs, membership tests and construction overhead become negligible compared to the rest of the symbolic part of the computation time.
In this sense, MDDs solve the problem of fast membership tests in monomial ideals during Gröbner basis computations.

The results highlight the prominence of the symbolic part of the sGB algorithm in \emph{AlgebraicSolving.jl}, with many examples where it dominates the total time, and no example below 50\% (without MDDs).
With MDDs, membership testing ceases to be a bottleneck, but it is only a fraction of the symbolic workload in which monomials are processed.
We expect further gains by designing similar data structures for the remaining
  symbolic operations.

\begin{acks}
  {\small This work has been supported by the
    \grantsponsor{erc}{European Research Council
      (ERC)}{https://erc.europa.eu} under the European Union’s Horizon
    Europe research and innovation programme, grant agreement
    \grantnum{erc}{101040794} (10000~DIGITS), by the \grantsponsor{fwo}{Flanders Research
    Foundation (FWO)}{https://www.fwo.be/en/}, grant agreements \grantnum{fwo}{G0F5921N} (Odysseus)
    and \grantnum{fwo}{G023721N}, and by \grantsponsor{kul}{KU Leuven}{https://www.kuleuven.be}, grant agreement \grantnum{kul}{iBOF/23/064}.}
\end{acks}

\begin{figure*}[p]
  \centering

  \includegraphics[width=\linewidth, trim={1.5cm 0cm 1.5cm 0cm}]{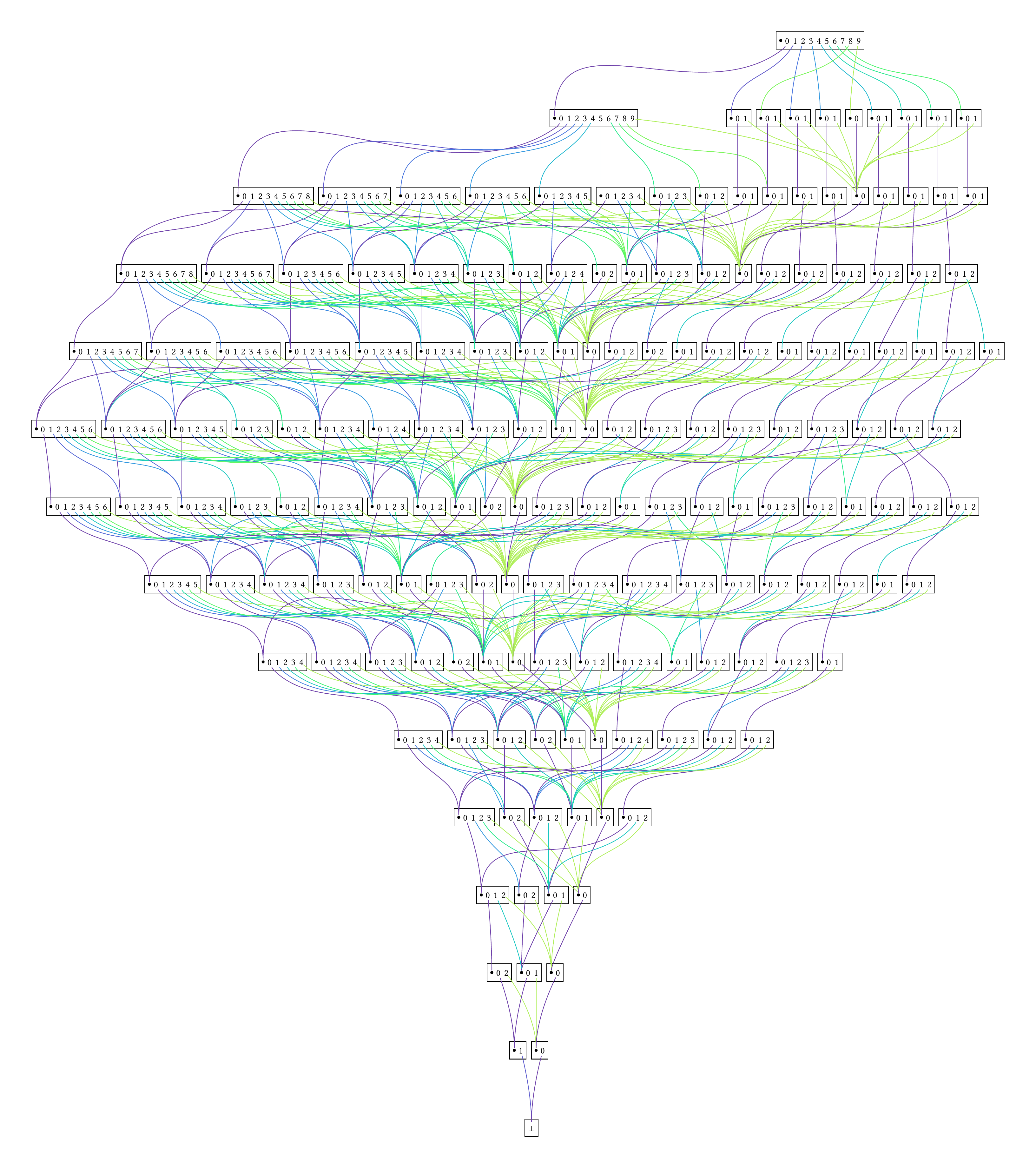}
  \caption{Monomial divisibility diagram of the monomial ideal in~$\mathbb{N}^{14}$ generated by the leading monomials of a Gröbner basis of \emph{eco-14}. The ideal is minimally generated by 2852~monomials. The MDD has 173 nodes.}
  \label{fig:eco14}
\end{figure*}

\printbibliography

\end{document}